\documentclass[conference]{IEEEtran}
\IEEEoverridecommandlockouts

\usepackage{graphicx}
\usepackage{amsmath}
\usepackage{amssymb}
\usepackage{amsthm}
\usepackage{tikz}
\usepackage[compress]{cite}

\newtheorem{assumption}{Assumption}
\newtheorem{proposition}{Proposition}
\newtheorem{definition}{Definition}

\newtheorem{corollary}{Corollary}

\DeclareMathOperator{\trace}{Tr}

\renewcommand{\S}{\mathcal{S}}

\title{Privacy-Preserving State Estimation with Crowd Sensors: An Information-Theoretic Respective}
\author{Farhad Farokhi\thanks{F.~Farokhi is with the Department of Electrical and Electronic Engineering at the University of Melbourne. e-mail: farhad.farokhi@unimelb.edu.au}}

\begin{document}

\maketitle

\begin{abstract}
    Privacy-preserving state estimation for linear time-invariant dynamical systems with crowd sensors is considered. At any time step, the estimator has access to measurements from a randomly selected sensor from a pool of sensors with pre-specified models and noise profiles. A Luenberger-like observer is used to fuse the  measurements with the underlying model of the system to recursively generate the state estimates. An additive privacy-preserving noise is used to constrain information leakage. Information leakage is measured via mutual information between the identity of the sensors and the state estimate conditioned on the actual state of the system. This captures an omnipotent adversary that not only can access state estimates but can also gather direct high-quality state measurements. Any prescribed level of information leakage is shown to be achievable by appropriately selecting the variance of the privacy-preserving noise. Therefore, privacy-utility trade-off can be fine-tuned.
\end{abstract}

\section{Introduction}
Crowd sensing is an emerging technique for estimation of variables of interest, such as traffic and pollution, using individuals' mobile devices capable of sensing and computing (e.g., smart phones), and has risen in popularity given prevalence of Internet of Things (IoT) devices and smart phones in everyday life~\cite{6069707,dutta2017towards, yan2017cloud}. Crowd sensors take local measurements that are fused together to construct accurate estimates of variables. However, the devices and their users can unintentionally leave fingerprints in the estimated values. For instance, if we acquire accurate measurements of traffic in a specific location, we can inevitably infer that a crowd sensing road user has been there. These information can be stitched together to infringe on privacy or security\footnote{For instance, aggregative heatmaps generated by Strava, a fitness tracking app, revealed the location of a secret US army base~\cite{gaurdian_strava}} of crowd sensing users. This has motivated development and analysis of private crowd sensing for data gathering and estimation~\cite{farokhi2016preserving, boutsis2013privacy, 7054716, huang2012privacy, 6480865, jin2016enabling}.
 
Privacy definition and analysis broadly fall within two category of differential privacy~\cite{dwork2006calibrating, dwork2014algorithmic} and information-theoretic privacy~\cite{wang2016relation,diaz2019robustness,issa2019operational}. Note that a third category based on anonymization, binning, and obfuscation has  been historically present in data science and statistics, but does not enjoy the strong guarantees of differential privacy and information-theoretic privacy~\cite{li2006t,sweeney2002k}. Differential privacy requires that reported outputs are relatively insensitive  to the data of any single individual; the extent of this insensitivity is reflected in a design parameter called privacy budget. This is often achieved by the use of additive noise~\cite{dwork2014algorithmic}. Information-theoretic privacy however focuses on systematically measuring private information leakage and developing optimal policies to constraint the leakage. A major difficulty with data privacy in time series or dynamical environments is the accumulation of privacy leakage over time, e.g., referred to as composition in differential privacy~\cite{dwork2014algorithmic}. This requires us to weaken the privacy guarantees with time~\cite{NIPS2013_c850371f}, to discount distant events~\cite{farokhi2020temporally}, to restrict number of releases~\cite{blocki2016differentially}, or to account for per-step information leakage~\cite{koufogiannis2017differential}.  This paper takes an information-theoretic path to dynamic privacy in crowd sensing.

In this paper, particularly, we consider privacy-preserving state estimation for linear time-invariant dynamical systems with crowd sensors. Each sensor has a specific model (what it measures) and noise profile (variance of measurement noise). At any given time, a single sensor is selected at random from a pool of heterogeneous sensors and provides an output measurement to the estimator. The estimator then uses a Luenberger-like observer to fuse these state measurements with the underlying model of the system to generate the state estimates in real time. The estimator may also add some noise to the state estimate to mask the identity of the sensors contributing to the estimation. Note that the estimated state or its accuracy can be used to identify the participating sensors. For instance, if two sensors one with large noise (e.g., using an outdated equipment) and one with small noise (e.g., using state-of-the-art equipment) contribute measurements to state estimation, we can identify the time instances in which the accurate sensor is used based on the quality of the estimation. This can be done by comparing the estimate with ground truth (when the adversary is omnipotent) or using the covariance matrix of the estimator (if shared by the estimator). We  measure private information leakage regarding the sequence of sensors used via mutual information between identity of the sensors and the state estimate conditioned on the actual state of the system. By conditioning on the actual state, we are modeling a very powerful adversary that not only can access the state estimates but can also gather direct high-quality state measurements. This would provide an upper bound on the actual information leakage in any weaker alternative scenario (i.e., realistic experimental scenarios) and is thus useful for privacy analysis in adversarial or safety-critical settings. Variance of the state estimate also enters the problem formulation as a measure of utility. We provide a bound for the measure of information leakage as a function of the sensor characteristics, estimation quality, and variance of the additive privacy-preserving noise. We show that we can achieve any prescribed level of information leakage, i.e., requested privacy guarantee, by appropriately selecting the variance of the additive noise. We can therefore fine-tune privacy-utility trade-off using the additive privacy-preserving noise. 

The rest of the paper is organized as follows. This section finishes with a brief notation overview. We present the problem formulation with measures of information leakage and utility in Section~\ref{sec:problem}. The privacy analysis and privacy-utility trade-off are presented in Section~\ref{sec:result}. A numerical example to illustrate the results is presented in Section~\ref{sec:numerical}. Finally, Section~\ref{sec:conc} presents some concluding remarks and avenues for future research.

\subsection{Notation}
The sets of real, integer, and natural numbers are denoted by $\mathbb{R}$, $\mathbb{Z}$, and $\mathbb{N}$, respectively. Let $\mathbb{N}_0=\mathbb{N}\cup\{0\}$. 
We write $X\succ 0$ ($X\succeq 0$) if $X$ is a symmetric positive definite (semi-definite) matrix. For any sequence of variables $x[n],\dots,x[m]\in \mathbb{X}$ with $m\geq n$, we use the notation $x[n:m]=(x[n],\dots,x[m])\in\mathbb{X}^{m-n+1}$, where $\mathbb{X}^{d}$ is the $d$-fold Cartesian product of the set $\mathbb{X}$ for any $d\in\mathbb{N}$.

\section{Problem Formulation} \label{sec:problem}
Consider linear time-invariant discrete-time system
\begin{align} \label{eqn:system}
    \mathbf{x}[k+1]=A\mathbf{x}[k]+\mathbf{w}[k],\quad  \forall k\in\mathbb{N}_0,
\end{align}
where $\mathbf{x}[k]\in\mathbb{R}^n$ is the state and $\mathbf{w}[k]\in\mathbb{R}^n$ is the process noise. 

\begin{assumption} \label{assum:process_noise} The process noise $(\mathbf{w}[k])_{k\in\mathbb{N}_0}$ is a sequence of identically and independently distributed (i.i.d.) zero-mean Gaussian random variables with covariance $W\succeq 0$. 
\end{assumption}

\begin{assumption} \label{assum:inti_condition} The initial condition $\mathbf{x}[0]$ is a zero-mean Gaussian random variable with covariance $X_0\succeq 0$. 
\end{assumption}

The state of the system in~\eqref{eqn:system} is measured by a set of sensors $\S:=\{1,\dots,m\}$. At any given time, the state estimator has access to a measurement from a randomly selected sensor. This can be viewed as an abstraction of crowd-sensing, i.e., a group of sensors provides state measurements at various time instants, one at a time, to the operator. For instance, the sensors could be vehicles traveling over a transportation network, where each one provides a measurement of the traffic flow in their vicinity. Let $\mathbf{s}[k]\in \S$ denote the identity of the sensor that provides a measurement at time $k\in\mathbb{N}_0$, i.e., $\mathbf{s}[k]=i$ if sensor $i\in \S$ provides the state measurement at time $k$. If $\mathbf{s}[k]=i$, we have access to state measurements of the form:
\begin{align} \label{eqn:output}
    \mathbf{y}[k]=C_{i}\mathbf{x}[k]+\mathbf{v}_{i}[k],
\end{align}
where $\mathbf{y}[k]\in\mathbb{R}^{p_i}$ is the sensing output and $\mathbf{v}_{i}[k]\in\mathbb{R}^{p_i}$ is the measurement noise.

\begin{assumption} \label{assum:meas_noise}
    For each sensor $i\in \S$, the measurement noise $(\mathbf{v}_i[k])_{k\in\mathbb{N}_0:\mathbf{s}[k]=i}$ is a sequence of i.i.d. zero-mean Gaussian random variables with covariance $V_i\succeq 0$. 
\end{assumption}

The noise in~\eqref{eqn:output}, formalized in Assumption~\ref{assum:meas_noise}, can be caused by instrumentation inaccuracies or can be artificially added for privacy-preserving purposes. 

\begin{assumption} \label{assum:sensor_select}
    The sensor selection strategy is i.i.d., i.e., for all $s,s[0],\dots,s[k-1]\in\S$,
    \begin{align}
        \mathbb{P}\{\mathbf{s}[k]\!=\!s\,|\,\mathbf{s}[0\!:\!k\!-\!1]\!=\!s[0\!:\!k\!-\!1]\}\!=\!\mathbb{P}\{\mathbf{s}[k]\!=\!s\}\!=\!p(s).
    \end{align}
\end{assumption}

Based on the measurements in~\eqref{eqn:output}, we can construct Luenberger-like state estimator of the form:
\begin{align} \label{eqn:estimator}
    \mathbf{\hat{x}}[k+1]=A\mathbf{\hat{x}}[k]+L(\mathbf{y}[k]-C_{\mathbf{s}[k]}\mathbf{\hat{x}}[k])+\boldsymbol{\xi}[k],
\end{align}
where $\boldsymbol{\xi}[k]\in\mathbb{R}^{n}$ is a privacy-preserving noise that the system operator can add to its state estimate to protect the identity of the sensors used so far, i.e., $\mathbf{s}[0:k]$, for constructing the state estimate. 

\begin{assumption}
    The privacy-preserving noise $(\boldsymbol{\xi}[k])_{k\in\mathbb{N}_0}$ is a sequence of i.i.d. zero-mean Gaussian random variables with covariance $\Xi\succeq 0$. 
\end{assumption}

The estimation error is then given by
\begin{align*}
    \mathbf{e}[k+1]
    :=&\mathbf{\hat{x}}[k+1]-\mathbf{x}[k+1]\\
    =&(A-LC_{\mathbf{s}[k]})\mathbf{e}[k]+L\mathbf{v}_{\mathbf{s}[k]}[k]+\boldsymbol{\xi}[k]-\mathbf{w}[k].
\end{align*}
The performance of the estimator in~\eqref{eqn:estimator} can be measured by
\begin{align} \label{eqn:performance_estimate}
    \mathfrak{P}(\Xi):=\lim_{k\rightarrow \infty}\frac{1}{T+1} \sum_{k=0}^{T} \mathbb{E}\{\mathbf{e}[k]^\top \Omega \mathbf{e}[k] \},
\end{align}
where $\Omega\succeq 0$ is a weighting matrix. Intuitively, the best estimation performance, i.e., smallest $\mathfrak{P}(\Xi)$, can be achieved by incorporating the smallest ``amount'' of noise, i.e., setting $\Xi=0$. However, that would result in a potentially larger privacy leakage, i.e., incorporating no privacy-preserving noise makes it easier to infer $\mathbf{s}[0:k]$ from $\mathbf{\hat{x}}[0:k]$ and $\mathbf{x}[0:k]$. To formalize this, we need to define a measure of private information leakage:
\begin{align} \label{eqn:private_info}
    \mathfrak{I}(\Xi):= \lim_{k\rightarrow \infty}\frac{1}{k+1} I(\mathbf{s}[0:k];\mathbf{\hat{x}}[0:k]|\mathbf{x}[0:k]),
\end{align}
where $I(\cdot,\cdot|\cdot)$ is the conditional mutual information~\cite[p.\,251]{coverelements}. Mutual information has been widely used in the privacy literature as a measure of private information leakage~\cite{makhdoumi2014information, farokhi2016privacy, murguia2021privacy, 1010079783031}. 

\begin{definition}[$\epsilon$-Private Estimation] The estimator in~\eqref{eqn:estimator}, with privacy-preserving noise $(\boldsymbol{\xi}[k])_{k\in\mathbb{N}_0}$, is $\epsilon$-private if $\mathfrak{I}(\Xi)\leq \epsilon$.    
\end{definition}

In this paper, our primary objective is to develop privacy-preserving state estimation policies, by computing covariance of privacy-preserving noise $\Xi$, to achieve $\epsilon$-privacy for all $\epsilon>0$. As a secondary objective, we would like to drive utility-privacy trade-off. We will investigate these problems in the next section. 

\section{Results} \label{sec:result}
Our first result is regarding the performance of the state estimator in~\eqref{eqn:estimator} as a function of the privacy-preserving noise. 

\begin{proposition} \label{prop:upperbound:performance} 
The performance of the estimator in~\eqref{eqn:estimator} is 
\begin{align*}
\mathfrak{P}(\Xi)=\trace(\Omega E[k]),    
\end{align*}
where $E[k]:=\mathbb{E}\{\mathbf{e}[k]\mathbf{e}[k]^\top \}$ computed recursively as
\begin{align}
    E[k\!+\!1]=&\mathbb{E}\{(A\!-\!LC_{\mathbf{s}[k]})E[k] (A\!-\!LC_{\mathbf{s}[k]})^\top\}\nonumber\\
    &+L\overline{V}L^\top+\Xi+W,
\end{align}
with $\overline{V}=\sum_{s\in\mathcal{S}}V_sp(s)$. Furthermore, if $A\!-\!LC_{i}$ are Schur matrices (i.e., all their eigenvalues reside within the unit disk) for all $i\in\S$, $\lim_{k\rightarrow \infty}E[k]=E^*$. 
\end{proposition}

\begin{proof}
See Appendix~\ref{proof:prop:upperbound:performance}.
\end{proof}

Computing an exact and explicit formula for information leakage is rather difficult due to the complicated nature of the underlying random variables (note the mixing of Gaussian variables entering the estimator due to random selection of sensors). Therefore, in the following proposition, we derive an upper bound for the information leakage. This bound can be used to derive sufficient conditions for achieving $\epsilon$-privacy. 

\begin{proposition} \label{prop:upperbound:information} 
Assume that $\lim_{k\rightarrow \infty}E[k]=E^*$. Then,
    \begin{align*}
        \mathfrak{I}(\Xi)
        \leq &\frac{1}{2}
        \ln(\det(L\overline{V}L^\top+\Xi
        +L\mathbb{E}\{\Delta C E^*\Delta C \}L^\top))\\
        &-\frac{1}{2}\sum_{s\in\mathcal{S}}p(s)\ln(\det(LV_{s}L^\top+\Xi))
    \end{align*}
    where $\overline{V}=\sum_{s\in\mathcal{S}}V_sp(s)$. 
\end{proposition}

\begin{proof}
See Appendix~\ref{proof:prop:upperbound:information}.
\end{proof}

The upper bound in Proposition~\ref{proof:prop:upperbound:information} can be used to show that $\epsilon$-privacy is achievable for all choices of $\epsilon>0$ by simply increasing the covariance of the privacy-preserving noise.  

\begin{corollary}\label{cor:achieve}
Assume that $\lim_{k\rightarrow \infty}E[k]=E^*$. Then, there exists $\Xi\succeq 0$ such that the estimator in~\eqref{eqn:estimator} is $\epsilon$-private for all $\epsilon>0$.
\end{corollary}

\begin{proof}
See Appendix~\ref{proof:cor:achieve}.
\end{proof}

In the next section, we demonstrate these results for a simple discrete-time system with two states. 


\begin{figure}
    \centering
    \begin{tikzpicture}
        \node[] at (0,0) {\includegraphics[width=1\linewidth]{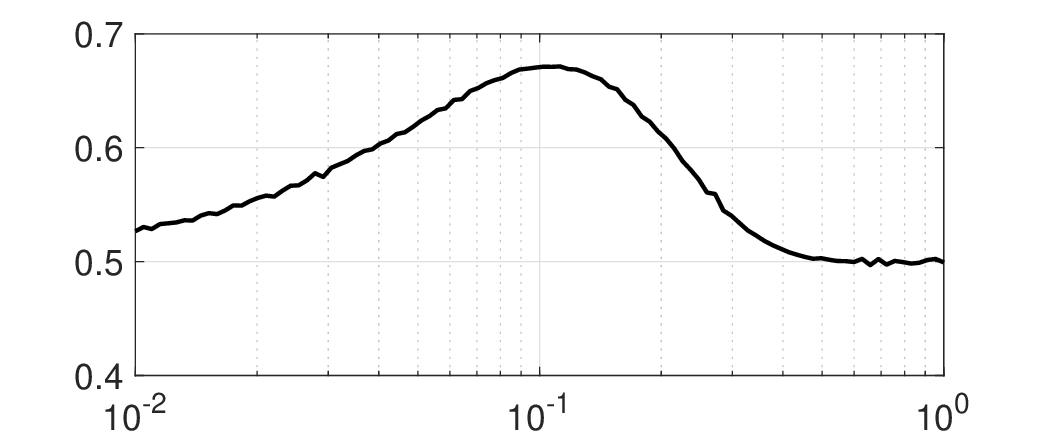}};
        \node[rotate=90] at (-4.1,0) {$\mathbb{P}\{\hat{\mathbf{s}}[k]=\mathbf{s}[k]\}$};
        \node[] at (0,-2.1) {$\tau$};
    \end{tikzpicture}
    \vspace*{-7mm}
    \caption{Probability of correctly detecting identity of the sensor at each time step $\mathbb{P}\{\hat{\mathbf{s}}[k]=\mathbf{s}[k]\}$ versus threshold $\tau$. }
    \label{fig:TP_vs_tau}

    \begin{tikzpicture}
        \node[] at (0,0) {\includegraphics[width=1\linewidth]{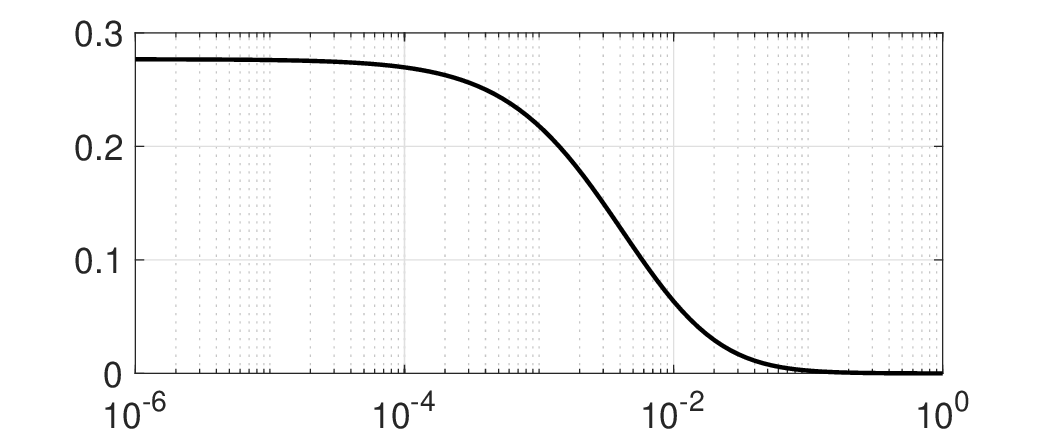}};
        \node[rotate=90] at (-4.1,0) {upper bound for $\mathfrak{I}(\Xi)$};
        \node[] at (0,-2.1) {$\sigma_\xi$}; 
    \end{tikzpicture}
    \vspace*{-7mm}
    \caption{Upper bound for private information leakage $\mathfrak{I}(\Xi)$ developed in Proposition~\ref{prop:upperbound:information} versus magnitude of privacy-preserving noise $\sigma_\xi$.}
    \label{fig:upper_bound}
\end{figure}

\section{Numerical Illustration} \label{sec:numerical}
In this section, we illustrate the results of the paper on a simple linear time-invariant discrete-time system modeling temperature dynamics of two interconnected rooms. Define
\begin{align*}
    \mathbf{x}[k]
    =
    \begin{bmatrix}
        T_1[k]-T_{\rm out}[k]
        \\
        T_2[k]-T_{\rm out}[k]
    \end{bmatrix},
\end{align*}
where $T_i[k]$ is the temperature in room $i=1,2$ and $T_{\rm out}[k]$ is the outdoor temperature. Each room exchanges heat with the other room and with the outside environment. Additionally, there are exogenous heat inputs (due to, e.g., people moving in and out of the rooms) that are modeled by process noise $\mathbf{w}[k]$ with zero mean and covariance $W=10^{-4} I$. Heat transfer is governed by linear resistive coupling, that is, heat flows proportional to temperature differences and inversely proportional to thermal resistances between rooms and outside. Thermal capacitance of each room determines the energy needed to change its temperature. As an example, the dynamics can be described by the discrete-time linear system in~\eqref{eqn:system} with 
\begin{align*}
A
=
\begin{bmatrix}
0.991 & 0.0075
\\
0.006 & 0.990
\end{bmatrix}.
\end{align*}
We assume that we can select uniformly at random from a pool of two sensors (the crowd sensors) with models
\begin{align*}
    C_{S_1}=C_{S_2}=C=\begin{bmatrix}
        1 & 0
    \end{bmatrix},
    V_1=10^{-1}, V_2=10^{-2}.
\end{align*}
We use the estimator in~\eqref{eqn:estimator} with 
\begin{align*}
    L=\begin{bmatrix}
        0.5 & 0
    \end{bmatrix}^\top , \quad 
    \Xi=\begin{bmatrix}
        \sigma_\xi & 0 \\ 0 & 10^{-32}
    \end{bmatrix},
\end{align*}
where we can select $\sigma_\xi$ to attain a certain level of privacy. Note that we do not need to add much noise to the estimate of the second state as it is not directly measured and thus does not reveal much about the identity of the sensor. Nonetheless, we use a non-zero covariance for the second state to avoid numerical issues with the upper bound developed in Proposition~\ref{prop:upperbound:information} (as otherwise the determinant of the matrices will be equal to zero and we end up with logarithm of zero being subtracted from logarithm of zero, which is ill-defined). Note that sensor $2$ is far superior to sensor $1$ (as $V_2/V_1=0.1\ll 1$). This information can be used by an adversary to develop a threshold-based strategy for identifying the sensor providing data at any given time. We can particularly adopt the following estimator for an adversary:
\begin{align} \label{eqn:policy_adv}
    \hat{\mathbf{s}}[k]=
    \begin{cases}
        1, & |C(\mathbf{x}[k]-\hat{\mathbf{x}}[k])|\geq \tau,\\
        2, & \mbox{otherwise}.
    \end{cases}
\end{align}
To select an appropriate threshold, we can select $\Xi=0$ and observe the probability of successful detection as a function of the threshold $\tau$. Figure~\ref{fig:TP_vs_tau} illustrates the probability of correctly detecting identity of the sensor at each time step $\mathbb{P}\{\hat{\mathbf{s}}[k]=\mathbf{s}[k]\}$, computed empirically across 100 runs, versus threshold $\tau$. Note that this probably does not become smaller than $1/2$, which is the success rate of purely guessing the identity of the sensor (the so-called dart throwing monkey). At $\tau=10^{-1}$ (which is nearly the optimal threshold), the policy in~\eqref{eqn:policy_adv} is accurate $2/3$ of the times, which is more than 30\% better than the baseline of purely guessing. Noting the simplicity of the adversarial estimation policy in~\eqref{eqn:policy_adv}, this is a remarkable feat. An interesting avenue for future research is to develop superior policies for identifying identity of the sensors using ideas from statistics and machine learning. 

\begin{figure}
    \begin{tikzpicture}
        \node[] at (0,0) {\includegraphics[width=1\linewidth]{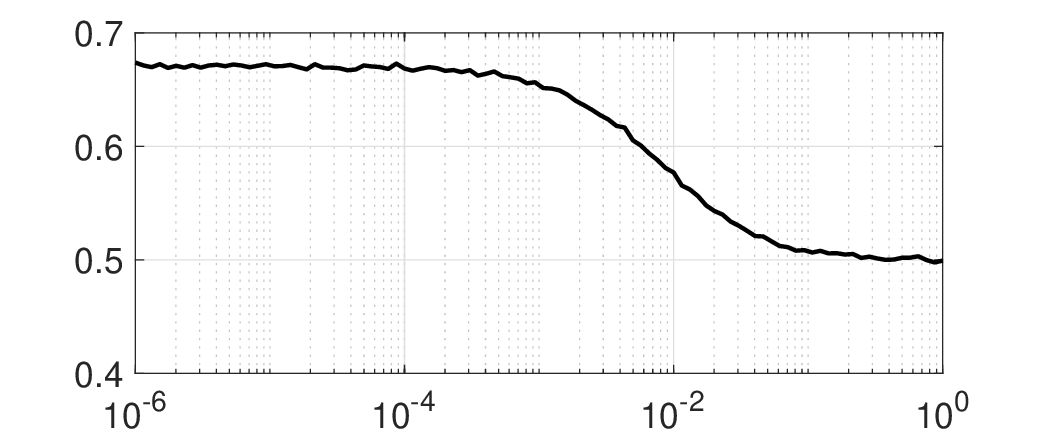}};
        \node[rotate=90] at (-4.1,0) {$\mathbb{P}\{\hat{\mathbf{s}}[k]=\mathbf{s}[k]\}$};
        \node[] at (0,-2.1) {$\sigma_\xi$}; 
    \end{tikzpicture}
    \vspace*{-7mm}
    \caption{Probability of correctly detecting identity of the sensor at each time step $\mathbb{P}\{\hat{\mathbf{s}}[k]=\mathbf{s}[k]\}$ versus magnitude of privacy-preserving noise $\sigma_\xi$.}
    \label{fig:TP_vs_noise}

    \begin{tikzpicture}
        \node[] at (0,0) {\includegraphics[width=1\linewidth]{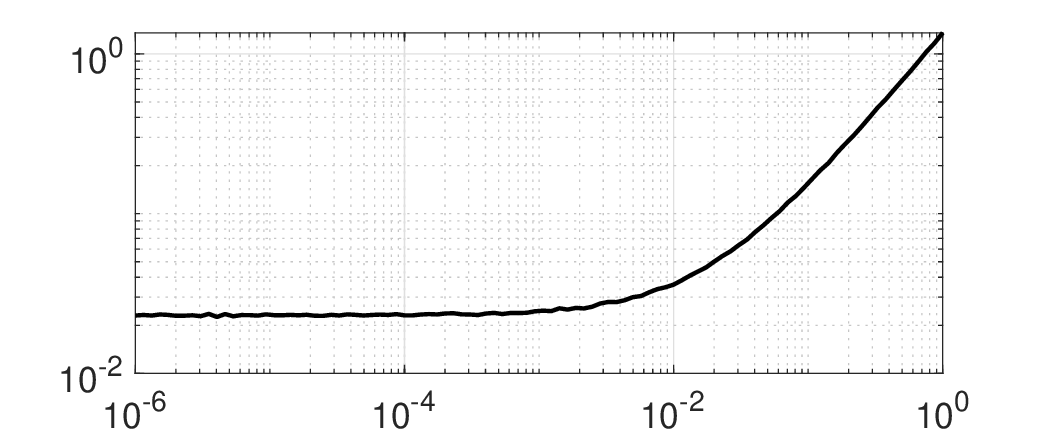}};
        \node[rotate=90] at (-4.1,0) {$\mathbb{E}\{\|\mathbf{x}[k]-\hat{\mathbf{x}}[k]\|_2^2\}$};
        \node[] at (0,-2.1) {$\sigma_\xi$};
    \end{tikzpicture}
    \vspace*{-7mm}
    \caption{State estimation error $\mathbb{P}\{\|\mathbf{x}-\hat{\mathbf{x}}\|_2^2\}$  versus magnitude of privacy-preserving noise $\sigma_\xi$.}
    \label{fig:error_vs_noise}
\end{figure}

Earlier, it was proved that the additive privacy-preserving noise $\boldsymbol{\xi}[k]$ can reduce the information leakage to be within any arbitrary range; see Corollary~\ref{cor:achieve}. This is demonstrated in the remainder of this section. Figure~\ref{fig:upper_bound} shows the upper bound for private information leakage $\mathfrak{I}(\Xi)$ developed in Proposition~\ref{prop:upperbound:information} versus magnitude of privacy-preserving noise, measured by its variance $\sigma_\xi$. This clear aligns with our analysis demonstrating that we can reduce the information leakage by increasing the magnitude of the noise. The effect of the noise can also be investigated empirically on the success the developed adversarial sensor estimation policy in~\eqref{eqn:policy_adv}. Figure~\ref{fig:TP_vs_noise} illustrates the probability of correctly detecting identity of the sensor at each time step $\mathbb{P}\{\hat{\mathbf{s}}[k]=\mathbf{s}[k]\}$ versus magnitude of privacy-preserving noise $\sigma_\xi$. Figures~\ref{fig:upper_bound} and~\ref{fig:TP_vs_noise} show remarkably similar trends. However, privacy preservation often comes at a cost. This can be seen through the effect of the privacy-preserving noise on the estimation error. Figure~\ref{fig:error_vs_noise} shows the state estimation error $\mathbb{E}\{\|\mathbf{x}[k]-\hat{\mathbf{x}}[k]\|_2^2\}$ versus magnitude of privacy-preserving noise $\sigma_\xi$. Figures~\ref{fig:TP_vs_noise} and~\ref{fig:error_vs_noise} illustrate the privacy-utility trade-off. As the magnitude of the privacy-preserving noise increases, the adversary would have a harder time to identify the identity of the sensors contributing to the state estimation but the estimation error also worsens. 

\section{Conclusions and Future Work}
\label{sec:conc}
We considered privacy-preserving state estimation for linear time-invariant dynamical systems with crowd sensors. Crowd sensors were modeled as a group of sensors with varying models that can be sampled randomly.  We used an additive privacy-preserving noise within a Luenberger-type observer to minimize information leakage, measured using mutual information. 
The results were demonstrated on a small numerical example. Future work can focus on experimental verification of the results and extension to nonlinear dynamics. 

\appendices

\section{Proof of Proposition~\ref{prop:upperbound:performance}}
\label{proof:prop:upperbound:performance}
Evidently, $\mathbb{E}\{\mathbf{e}[k]\}=0$ for all $k\in\mathbb{N}_0$. Therefore, 
\begin{align*}
    E[k+1]
    \!:=&\mathbb{E}\{\mathbf{e}[k+1]\mathbf{e}[k+1]^\top \}\\
    =&\mathbb{E}\{((A-LC_{\mathbf{s}[k]})\mathbf{e}[k]\!+\!Lv_{\mathbf{s}[k]}[k]\!+\!\boldsymbol{\xi}[k]\!-\!\mathbf{w}[k])\\
    &\times\!((A-LC_{\mathbf{s}[k]})\mathbf{e}[k]\!+\!Lv_{\mathbf{s}[k]}[k]\!+\!\boldsymbol{\xi}[k]\!-\!\mathbf{w}[k])^\top\! \}\\
    =& \mathbb{E}\{(A-LC_{\mathbf{s}[k]})E[k] (A-LC_{\mathbf{s}[k]})^\top\}\\
    &+L\overline{V}L^\top+\Xi+W,
\end{align*}
where $\overline{V}=\mathbb{E}\{v_{\mathbf{s}[k]}[k]v_{\mathbf{s}[k]}[k]^\top\}=\sum_{i=1}^m \mathbb{P}\{\mathbf{s}[k]=i\}V_{i}=\sum_{i=1}^m p(s)V_{i}$. Finally, note that the mapping $E\mapsto \mathbb{E}\{(A-LC_{\mathbf{s}[k]})E[k] (A-LC_{\mathbf{s}[k]})^\top\}$ is contractive if $A\!-\!LC_{i}$ are Schur matrices for all $i\in\S$. Therefore, the convergence immediately follows from the Banach fixed point theorem~\cite[p.\,3]{pata2019fixed}.

\section{Proof of Proposition~\ref{prop:upperbound:information}}
\label{proof:prop:upperbound:information}
    The chain rule for mutual information~\cite[Theorem~8.6.2]{coverelements} gives
    \begin{align}
        I(\mathbf{s}[0:k];&\mathbf{\hat{x}}[0:k]|\mathbf{x}[0:k])\nonumber\\
        =& I(\mathbf{s}[0:k];\mathbf{\hat{x}}[0]|\mathbf{x}[0:k])\nonumber\\
        &+I(\mathbf{s}[0:k];\mathbf{\hat{x}}[1]|\mathbf{\hat{x}}[0],\mathbf{x}[0:k])\nonumber\\
        &+I(\mathbf{s}[0:k];\mathbf{\hat{x}}[2]|\mathbf{\hat{x}}[0:1],\mathbf{x}[0:k])\nonumber\\
        &+\dots\nonumber\\
        &+I(\mathbf{s}[0:k];\mathbf{\hat{x}}[k]|\mathbf{\hat{x}}[0:k-1],\mathbf{x}[0:k]).\label{eqn:proof:1}
    \end{align}
    For any $0\leq t\leq k$, we have
    \begin{align}
        I(\mathbf{s}[0:k];&\mathbf{\hat{x}}[t]|\mathbf{\hat{x}}[0:t-1],\mathbf{x}[0:k])\nonumber\\
        =&h(\mathbf{\hat{x}}[t]|\mathbf{\hat{x}}[0:t-1],\mathbf{x}[0:k])\nonumber\\
        &-h(\mathbf{\hat{x}}[t]|\mathbf{\hat{x}}[0:t-1],\mathbf{s}[0:k],\mathbf{x}[0:k])\nonumber\\
        =&h(\mathbf{\hat{x}}[t]|\mathbf{\hat{x}}[t-1],\mathbf{x}[t-1])\nonumber\\
        &-h(\mathbf{\hat{x}}[t]|\mathbf{\hat{x}}[t-1],\mathbf{s}[t-1],\mathbf{x}[t-1])\nonumber\\
        =&I(\mathbf{s}[t-1];\mathbf{\hat{x}}[t]|\mathbf{\hat{x}}[t-1],\mathbf{x}[t-1]),
        \label{eqn:proof:2}
    \end{align}
    where the first and the last equality follow from the relationship between mutual information and entropy~\cite[p.\,251]{coverelements}, and the second equality follows from the specific structure of the estimator in~\eqref{eqn:estimator} and the i.i.d. nature of the sensor selections in Assumption~\ref{assum:sensor_select}. Also note that $I(\mathbf{s}[0:k];\mathbf{\hat{x}}[0]|\mathbf{x}[0:k])=0$ because $\mathbf{s}[0:k]$ and $\mathbf{\hat{x}}[0]$ are statistically independent. Combining~\eqref{eqn:proof:1} and~\eqref{eqn:proof:2} shows that
    \begin{align}
    I(\mathbf{s}[0:k]&;\mathbf{\hat{x}}[0:k]|\mathbf{x}[0:k])\nonumber
        \\=& \sum_{t=1}^{k} I(\mathbf{s}[t-1];\mathbf{\hat{x}}[t]|\mathbf{\hat{x}}[t-1],\mathbf{x}[t-1]).
    \end{align}
    Now, we focus on developing an upper bound for each term  $I(\mathbf{s}[t-1];\mathbf{\hat{x}}[t]|\mathbf{\hat{x}}[t-1],\mathbf{x}[t-1])$. To do so, note that
    $I(\mathbf{s}[t-1];\mathbf{\hat{x}}[t]|\mathbf{\hat{x}}[t-1],\mathbf{x}[t-1])
        =h(\mathbf{\hat{x}}[t]|\mathbf{\hat{x}}[t-1],\mathbf{x}[t-1])
        -h(\mathbf{\hat{x}}[t]|\mathbf{\hat{x}}[t-1],\mathbf{s}[t-1],\mathbf{x}[t-1]).$
    Therefore, upper bounding $I(\mathbf{s}[t-1];\mathbf{\hat{x}}[t]|\mathbf{\hat{x}}[t-1],\mathbf{x}[t-1])$ can be achieved for finding an upper bound for $h(\mathbf{\hat{x}}[t]|\mathbf{\hat{x}}[t-1],\mathbf{x}[t-1])$ and a lower bound for $h(\mathbf{\hat{x}}[t]|\mathbf{\hat{x}}[t-1],\mathbf{s}[t-1],\mathbf{x}[t-1])$. Let us start with finding a lower bound for $h(\mathbf{\hat{x}}[t]|\mathbf{\hat{x}}[t-1],\mathbf{s}[t-1],\mathbf{x}[t-1])$. Note that, because $\mathbf{\hat{x}}[k]=A\mathbf{\hat{x}}[k-1]+L(C_{\mathbf{s}[k-1]}\mathbf{x}[k-1]+\mathbf{v}_{\mathbf{s}[k-1]}[k-1]-C_{\mathbf{s}[k-1]}\mathbf{\hat{x}}[k-1])+\boldsymbol{\xi}[k-1]$, we get 
    \begin{align*}
    h(\mathbf{\hat{x}}[t]|\mathbf{\hat{x}}[t-1],&\mathbf{s}[t-1]\!=\!s,\mathbf{x}[t-1])\\
    &=h(L\mathbf{v}_{s}[t-1]+\boldsymbol{\xi}[t-1])\\
    &=\frac{1}{2}\ln((2\pi e)^{n}\det(LV_{s}L^\top+\Xi)),
    \end{align*}
    and, as a result,
    \begin{align*}
    h(\mathbf{\hat{x}}&[t]|\mathbf{\hat{x}}[t-1],\mathbf{s}[t-1],\mathbf{x}[t-1])
    \\
    &= \sum_{s\in\mathcal{S}}p(s)h(\mathbf{\hat{x}}[t]|\mathbf{\hat{x}}[t-1],\mathbf{s}[t-1]=s,\mathbf{x}[t-1])\\
    &=\frac{1}{2}\sum_{s\in\mathcal{S}}\ln((2\pi e)^{n}\det(LV_{s}L^\top+\Xi))p(s)\\
    &=\frac{1}{2}\ln((2\pi e)^{n})+\frac{1}{2}\sum_{s\in\mathcal{S}}p(s)\ln(\det(LV_{s}L^\top+\Xi)).
    \end{align*}
    Now, we focus on finding an upper bound for $h(\mathbf{\hat{x}}[t]|\mathbf{\hat{x}}[t-1],\mathbf{x}[t-1])$. To do so, note the inequality~\eqref{eqn:long:1}, on top of the next page.
    \begin{figure*}
    \begin{align}
    h(\mathbf{\hat{x}}[t]|\mathbf{\hat{x}}[t-1],\mathbf{x}[t-1])
    &=h(A\mathbf{\hat{x}}[t-1]+L(C_{\mathbf{s}[t-1]}\mathbf{x}[t-1]+\mathbf{v}_{\mathbf{s}[t-1]}[t-1]-C_{\mathbf{s}[t-1]}\mathbf{\hat{x}}[t-1])+\boldsymbol{\xi}[t-1]|\mathbf{\hat{x}}[t-1],\mathbf{x}[t-1])\nonumber
    \\
    &=h(LC_{\mathbf{s}[t-1]}\mathbf{x}[t-1]+L\mathbf{v}_{\mathbf{s}[t-1]}[t-1]+\boldsymbol{\xi}[k-1]-LC_{\mathbf{s}[t-1]}\mathbf{\hat{x}}[t-1]|\mathbf{\hat{x}}[t-1],\mathbf{x}[t-1]),
    \label{eqn:long:1}
    \end{align}
    \hrule 
    \end{figure*}
    Define $\mathbf{z}=L\mathbf{v}_{\mathbf{s}[t-1]}[t-1]+\boldsymbol{\xi}[k-1]-LC_{\mathbf{s}[t-1]}(\mathbf{\hat{x}}[t-1]-\mathbf{x}[t-1])$. We have
    \begin{align*}
        \overline{ \mathbf{z}}:=\mathbb{E}\{\mathbf{z}|\mathbf{\hat{x}}[t-1],\mathbf{x}[t-1]
        \}
        =L\overline{C}(\mathbf{\hat{x}}[t-1]-\mathbf{x}[t-1])
    \end{align*}
    where $\overline{C}=\mathbb{E}\{C_{\mathbf{s}[t-1]}\}$. Furthermore,
    \begin{align*}
        \mathbb{E}\{(\mathbf{z}&-\overline{\mathbf{z}})(\mathbf{z}-\overline{\mathbf{z}})^\top|\mathbf{\hat{x}}[t-1],\mathbf{x}[t-1]
        \}
        \\\leq& L\overline{V}L^\top+\Xi
        \\
        &+L\mathbb{E}\{\Delta C\mathbf{e}[t-1]\mathbf{e}[t-1]^\top \Delta C|\mathbf{e}[t-1] \}L^\top 
    \end{align*}
    where $\overline{V}=\sum_{s}V_sp(s)$ and $\Delta C=\overline{C}-C_{\mathbf{s}[t-1]}$.
    Noting that Gaussian random variables have the highest entropy among all random variables with a given covariance~\cite[]{}, we get
    \begin{align*}
        h(\mathbf{z}|&\mathbf{\hat{x}}[t-1],\mathbf{x}[t-1])
        \\
        \leq& \frac{1}{2}\ln((2\pi e)^{n})\\
        &+\frac{1}{2}
        \mathbb{E}\{
        \ln(\det(L\overline{V}L^\top+\Xi
        \\
        &\hspace{.5in}+L\mathbb{E}\{\Delta C\mathbf{e}[t-1]\mathbf{e}[t-1]^\top \Delta C|\mathbf{e}[t-1] \}L^\top))\}
        \\
        \leq& \frac{1}{2}\ln((2\pi e)^{n})\\
        &+\frac{1}{2}
        \ln(\det(L\overline{V}L^\top+\Xi
        \\
        &\hspace{.5in}+L\mathbb{E}\{\Delta C\mathbf{e}[t-1]\mathbf{e}[t-1]^\top \Delta C \}L^\top))
        \\
        \leq& \frac{1}{2}\ln((2\pi e)^{n})\\
        &+\frac{1}{2}
        \ln(\det(L\overline{V}L^\top+\Xi
        +L\mathbb{E}\{\Delta CE[t-1] \Delta C \}L^\top)).
    \end{align*}
    Combining all these inequalities results in
    \begin{align*}
        I(\mathbf{s}[0:k];&\mathbf{\hat{x}}[0:k]|\mathbf{x}[0:k])\\
        \leq &\frac{k}{2}
        \ln(\det(L\overline{V}L^\top+\Xi
        +L\mathbb{E}\{\Delta CE[t-1] \Delta C \}L^\top))\\
        &-\frac{k}{2}\sum_{s\in\mathcal{S}}p(s)\ln(\det(LV_{s}L^\top+\Xi)).
    \end{align*}
    This concludes the proof.

\section{Proof of Corollary~\ref{cor:achieve}}
\label{proof:cor:achieve}
Note that if $\lim_{\Xi=\lambda I, \lambda\rightarrow\infty} \mathfrak{I}(\Xi)= 0$, the statement of this corollary holds (sufficient result). Let $A=L\overline{V}L^\top+L\mathbb{E}\{\Delta C E^*\Delta C \}L^\top$ and $B_s=LV_{s}L^\top$. Hence,
\begin{align*}
    \mathfrak{I}(\Xi)
    \leq &\frac{1}{2}
    \ln(\det(A+\Xi))-\frac{1}{2}\sum_{s\in\mathcal{S}}p(s)\ln(\det(B_s+\Xi))\\
    =& \frac{1}{2}\sum_{s\in\mathcal{S}}p(s)
    \ln(\det((B_s+\Xi)^{-1}(A+\Xi)))\\
    =& \frac{1}{2}\sum_{s\in\mathcal{S}}p(s)
    \trace((B_s+\Xi)^{-1}(A+\Xi)-I)\\
    \leq & \frac{1}{2}
    \trace(\Xi^{-1}(A+\Xi)-I)
    \\
    \leq & \frac{1}{2}
    \trace(\Xi^{-1}A)
    \\
    \leq & \frac{\lambda^{-1}}{2}
    \trace(A).
\end{align*}
As a result, $\lim_{\Xi=\lambda I, \lambda\rightarrow\infty}\mathfrak{I}(\Xi)\leq 0$, which proves the result noting that $\mathfrak{I}(\Xi)\geq 0$.





\bibliography{ref}
\bibliographystyle{ieeetr}

\end{document}